\documentclass[UKenglish,cleveref, autoref, thm-restate]{lipics-v2021}

\usepackage{amsfonts,amstext,amsmath,amssymb,stmaryrd,mathrsfs,calc,amsthm,xspace,mathtools}
\usepackage{bm}% \bm{..} = bold math
\usepackage{comment}
\usepackage{framed}
\usepackage{graphics}
\usepackage{enumerate}
\usepackage{todonotes}
\usepackage{mathtools}
\usepackage{xspace,tabularx,multirow}
\usepackage{color}
\usepackage{colortbl}

\usepackage{macros}
\usepackage{hyperref}
\usepackage{cleveref}

\usepackage{lipsum}

\usepackage{cite}

\usepackage{macros}
\input{ABACABA-word-utils}

\usepackage{tikz}
\usepackage[normalem]{ulem}
\usepackage{ifthen}
\usetikzlibrary{arrows, automata,calc,shapes.arrows, backgrounds, fadings,intersections,  decorations.markings, positioning, snakes}
\usetikzlibrary{shadows}
\tikzfading[name=fade out, inner color=transparent!0,
  outer color=transparent!100]
\tikzset{
    every picture/.style={>=stealth,auto,node distance=0.5cm,}%double equal sign distance},
}
\tikzstyle{every state}=[
    inner sep=0pt,
    minimum size=0.05cm,
    fill= white,
    fill opacity=0.2,
    draw= black,
    text= black,
    text opacity=1,
    scale=0.5,
    ]
    
\hideLIPIcs

\title{Sem{\"e}nov Arithmetic, Affine {VASS}, and String Constraints}

\titlerunning{Sem{\"e}nov Arithmetic, Affine {VASS}, and String Constraints}

\author{Andrei Draghici}{Department of Computer Science, University of Oxford, Oxford, United Kingdom}{andrei.draghici@cs.ox.ac.uk}{}{}

\author{Christoph Haase}{Department of Computer Science, University of Oxford, Oxford, United Kingdom}{christoph.haase@cs.ox.ac.uk}{}{}

\author{Florin Manea}{Computer Science Department and Campus-Institut Data Science, G{\"o}ttingen University, Germany}{florin.manea@informatik.uni-goettingen.de}{}{}

\authorrunning{A.\ Draghici, C.\ Haase, and F.\ Manea} 
\Copyright{} 

 %TODO mandatory; please add comma-separated list of keywords

\category{} %optional, e.g. invited paper

\relatedversion{} %optional, e.g. full version hosted on arXiv, HAL, or other respository/website
\funding{
  This work is part of a project that has received
  funding from the European Research Council (ERC) under the European
  Union's Horizon 2020 research and innovation programme (Grant
  agreement No.~852769, ARiAT). Florin Manea's work was supported by a Heisenberg Grant, funded by the Deutsche Forschungsgemeinschaft (DFG, project 466789228). }%optional, to capture a funding statement, which applies to all authors. Please enter author specific funding statements as fifth argument of the \author macro.

\usepackage[firstpage]{draftwatermark}
\SetWatermarkText{\hspace*{6.9in}\raisebox{-2.25in}{
    \includegraphics[scale=0.06]{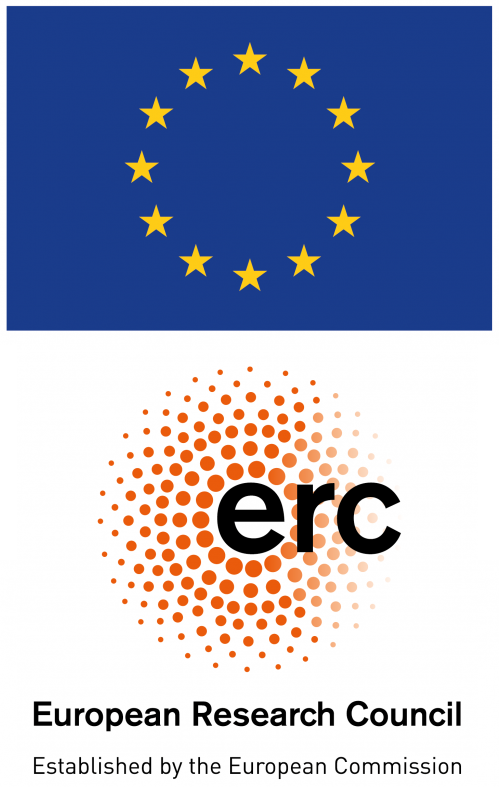}}}
\SetWatermarkAngle{0}

\acknowledgements{}%optional

\nolinenumbers %uncomment to disable line numbering

\EventEditors{John Q. Open and Joan R. Access}
\EventNoEds{2}
\EventLongTitle{42nd Conference on Very Important Topics (CVIT 2016)}
\EventShortTitle{CVIT 2016}
\EventAcronym{CVIT}
\EventYear{2016}
\EventDate{December 24--27, 2016}
\EventLocation{Little Whinging, United Kingdom}
\EventLogo{}
\SeriesVolume{42}
\ArticleNo{23}
\begin{document}

\DeclarePairedDelimiter{\ceil}{\lceil}{\rceil}
\DeclarePairedDelimiter\floor{\lfloor}{\rfloor}

%% \author{
%% \IEEEauthorblockN{1\textsuperscript{st} Andrei Draghici}
%% \IEEEauthorblockA{\textit{Department of Computer Science} \\
%% \textit{University of Oxford}\\
%% Oxford, United Kingdom \\
%% andrei.draghici@cs.ox.ac.uk}
%% \and
%% \IEEEauthorblockN{2\textsuperscript{nd} Christoph Haase}
%% \IEEEauthorblockA{\textit{Department of Computer Science} \\
%% \textit{University of Oxford}\\
%% Oxford, United Kingdom \\
%% christoph.haase@cs.ox.ac.uk}
%% \and 
%% \IEEEauthorblockN{3\textsuperscript{rd} Florin Manea}
%% \IEEEauthorblockA{\textit{Computer Science Department}\\ \textit{ and Campus-Institut Data Science} \\
%% \textit{G{\"o}ttingen University}\\
%% Germany \\}
%% }

%\author{Andrei Draghici}
%\author{Christoph Haase}
%\address{Department of Computer Science, University of Oxford, United Kingdom}
%\author{Florin Manea}
%\address{Computer Science Department and Campus-Institut Data Science, G{\"o}ttingen University, Germany}

\hyphenation{EXP-SPACE}

\maketitle

\begin{abstract}
  We study extensions of Sem{\"e}nov arithmetic, the first-order
  theory of the structure $\langle \N,+,2^x\rangle$. It is well-known
  that this theory becomes undecidable when extended with regular
  predicates over tuples of number strings, such as the B\"uchi
  $V_2$-predicate. We therefore restrict ourselves to the existential
  theory of Sem{\"e}nov arithmetic and show that this theory is
  decidable in EXPSPACE when extended with arbitrary regular
  predicates over tuples of number strings. Our approach relies on a
  reduction to the language emptiness problem for a restricted class
  of affine vector addition systems with states, which we show
  decidable in EXPSPACE. As an application of our result, we settle an
  open problem from the literature and show decidability of a class of
  string constraints involving length constraints.
\end{abstract}

\section{Introduction}

This paper studies the decidability and complexity of the existential
theory of an extension of the structure
$\langle \N,0,1,+,2^x \rangle$, where $2^x$ is the function mapping a
natural number $n$ to $2^n$.  Decidability of the first-order theory of
this structure was first shown by Sem{\"e}nov in a more general
framework using an automata-theoretic approach~\cite{Sem80}, and we
henceforth call this theory
\emph{Sem{\"e}nov arithmetic}. As shown by Cherlin and Point~\cite{CP86},
see also~\cite{Point10}, Sem{\"e}nov arithmetic admits quantifier
elimination and has a quantifier-elimination procedure that runs in
non-elementary time, and this upper bound is tight~\cite{Point10}. The
existential fragment of Sem{\"e}nov arithmetic has recently been shown
decidable in NEXP~\cite{BCM23} by giving a more elaborate quantifier
elimination procedure. Unlike its substructure Presburger arithmetic
(obtained from dropping the function $2^x$), Sem{\"e}nov arithmetic is
not automatic in the sense of the theory of automatic
structures~\cite{KN95,BG00}. The constant growth lemma~\cite{KM07}
states that for any function $f$ of an automatic structure,
$\abs{f(x_1,\ldots,x_n)} \le \abs{x_1} + \cdots + \abs{x_n} + c$ for
some constant $c$. This is clearly not the case for the function
$2^x$.

The decidability of Sem{\"e}nov arithmetic is fragile with respect to
extensions of the structure. For instance, it is not difficult to see
that extending Sem{\"e}nov arithmetic with the B\"uchi predicate
$V_2(x,y)$, where $V_2(x,y)$ holds whenever $x$ is the largest power
of two dividing $y$ without remainder, results in an undecidable
first-order theory, see e.g.~\cite{Point10}. However, this
undecidability result requires an $\exists^*\forall^*$-quantifier
prefix and does not rule out decidability of the existential
fragment. The main result of this paper is to show that the
existential theory of \emph{generalised Sem{\"e}nov arithmetic}, i.e.,
the existential theory of $\langle
\N,0,1,+,2^x,\{R_i\}_{i\ge 0} \rangle$, where $R_0,R_1,\ldots$ is an
enumeration of all regular languages over the alphabets $\{0,1\}^d$,
$d\ge 1$, is decidable in EXPSPACE. Non-automaticity of Sem{\"e}nov
arithmetic and undecidability of $\langle \N,0,1,+,2^x,V_2 \rangle$
rule out the possibility of approaching this existential theory via
automatic structures based on finite-state automata or via
quantifier-elimination \emph{\`a la} Cherlin and Point, since $V_2$ is
definable as a regular language over pairs of number strings.
Instead, our decidability result is based on a reduction to the
language emptiness problem of a special class of \emph{affine vector
addition systems with states (affine VASS)}.

A VASS comprises a finite-state controller with a finite number of
counters ranging over the natural numbers. In an affine VASS, when
taking a transition, every counter can be updated by applying an
affine function $x \mapsto a x + b$ to the current value,
provided that the resulting counter is non-negative. While
reachability in affine VASS is decidable for a single
counter~\cite{FGH13}, already in the presence of two counters
reachability becomes undecidable~\cite{Reichert2015}. Our reduction
consequently requires a restricted class of affine VASS to obtain
decidability. We call this class
\emph{restricted labelled affine VASS (restricted \laVASS)}. A
restricted \laVASS is an affine VASS with $d$ pairs of counters and
hence $2d$ counters in total. For every pair, the first counter does
not change until it keeps getting incremented at every transition; the
second counter is only updated via affine functions $x\mapsto 2
x$ and $x\mapsto 2 x+1$. A configuration consisting of a control
state and $2d$ counter values is accepting whenever the control state
is accepting and for every pair of counters, the first counter has the
same value as the second counter. We give an EXPSPACE procedure for
deciding emptiness of restricted \laVASS whose correctness proof is
based on a kind of counter elimination procedure in which we
successively encode counters into a finite state space while
preserving equi-non-emptiness. The tight syntactical restrictions
on \laVASS are necessary in order to obtain a decidable class of
affine VASS---relaxing those restrictions even only slightly leads to
undecidability of the language emptiness problem as we will later
discuss.

The EXPSPACE upper bound for existential generalised Sem{\"e}nov
arithmetic follows from a reduction to language non-emptiness of a
restricted \laVASS whose language encodes all solutions of the given
formula. Obtaining an elementary upper bound is difficult since it is
easily seen that smallest solutions of an existential formula of
Sem{\"e}nov arithmetic can be non-elementary in bit-length.

As an application of our EXPSPACE upper bound for existential
generalised Sem{\"e}nov arithmetic, we show that a certain class of
string constraints with length constraints is decidable in
EXPSPACE. It allows existentially quantifying over bit-strings, and to
assert that the value of a string variable lies in a regular language,
as well as Presburger-definable constraints over the lengths of the
bit-strings stored in string variables and the numerical values of
those variables (when viewed as encoding a number in binary).
Decidability of this class was left open in~\cite{BDGKMMN23}. We
settle this open problem by showing that it can be reduced to the
existential fragment of generalised Sem{\"e}nov arithmetic. Formulas
of this class of string constraints appear widely in practice---in
fact, essentially all formulas in the extensive collection of standard
real-world benchmark sets featured in~\cite{BDGKMMN23,BKMMDNG20} lie
in this class.

\section{Preliminaries}

% General notation
% Sem{\"e}nov arithmetic

\subsection{Basic notation}
By $\Z$ and $\N$ we denote the integers and non-negative integers,
respectively. Given an $m\times n$ integer matrix $A$, we denote by
$\norm{A}_{1,\infty}$ the $(1,\infty)$-norm of $A$, which is the
maximum over the sum of the absolute values of the coefficients of the
rows in $A$. For $\vec b\in \Z^m$, $\norm{\vec b}_\infty$ is the
largest absolute value of the numbers occurring in $\vec b$.

\subsection{Numbers as strings and strings as numbers}
Here and below, let $\Sigma=\{0,1\}$ be a binary alphabet. Any string
from $\Sigma^*$ has a natural interpretation as a binary encoding of a
natural number, possibly with an arbitrary number of leading
zeros. Conversely, any natural number in $\N$ can be converted into
its bit representation as a string in $\Sigma^*$. Finally, by
considering strings over $(\Sigma^k)^*$ for $k\ge 1$, we can represent
$k$-tuples of natural numbers as strings over $\Sigma^k$, and
\emph{vice versa}.

Formally, given $u=\vec u_n \vec u_{n-1}\dots \vec u_0 \in
(\Sigma^k)^*$, we define the tuple of natural numbers corresponding to
$u$ in \emph{most-significant digit first (msd)} notation as
\[
\eval{u} \defeq \sum_{i=0}^n 2^i \cdot \vec u_i\,.
\]
Note that $\eval{\cdot}$ is surjective but not injective.%%  Conversely,
%% for $\vec v\in \N^k \setminus \{ \vec 0 \}$, we define
%% $\ieval{n}\defeq u$ as the unique $u\in (\Sigma^k)^*$ such that
%% $\eval{u}=\vec v$ and $u\not\in \vec 0^+\cdot (\Sigma^k)^*$, i.e., the
%% msd-first encoding of $\vec v$ without any superfluous leading zeros.
%% Finally, $\ieval{\vec 0}\defeq \vec 0$.
We lift the definition of $\eval{\cdot}$
% and $\ieval{\cdot}$
to sets in the natural way.

\subsection{Generalised Sem{\"e}nov arithmetic}
For technical convenience, the structures we consider in this paper
are relational. We refer to the first-order theory of $\langle
\N,0,1,+,2^{(\cdot)} \rangle$ as \emph{Sem{\"e}nov arithmetic}, where
$+$ is the natural ternary addition relation, and $2^{(\cdot)}$ is the
power relation of base two, consisting of all tuples $(a,b) \in \N^2$
such that $b=2^a$. Sem{\"e}nov arithmetic is an extension of
Presburger arithmetic, which is the first-order theory of the
structure $\langle \N,0,1,+ \rangle$. It is known that Sem{\"e}nov
arithmetic is decidable and admits quantifier
elimination~\cite{Sem80,CP86,BCM23}.

For presentational convenience, atomic formulas of Sem{\"e}nov
arithmetic are one of the following:
\begin{itemize}
\item linear equations of the form $a_1 \cdot x_1 + \cdots a_d \cdot
  x_d = b$, $a_i,b \in \Z$, and
\item exponential equations of the form $x=2^y$.
\end{itemize}
Here, $x_1,\ldots,x_d,y$ are arbitrary first-order variables. Clearly,
richer atomic formulas such as $x + 2^{2^y} + y = z + 5$ can be
defined from those basic class of atomic formulas, since, in this
example, $x + 2^{2^y} + y = z + 5 \equiv \exists u\exists v\, u=2^v
\wedge v = 2^y \wedge x+u+y-z=5$. Moreover, since we are interpreting
numbers over non-negative integers, we can define the order relation
in existential Sem{\"e}nov arithmetic. This enables us to without loss
of generality assume that existential formulas of Sem{\"e}nov
arithmetic are positive, since $\neg (x=y) \equiv x<y \vee y<x$ and
$\neg(x=2^y) \equiv \exists z\, z=2^y \wedge \neg( x=z)$. 

The main contribution of this paper is to show that the existential
fragment of a generalisation of Sem{\"e}nov arithmetic is decidable.
Subsequently, we write $\vec 0$ to denote a tuple of 0s in any
arbitrary but fixed dimension. \emph{Generalised Sem{\"e}nov
  arithmetic} additionally allows for non-negated atomic formulas
$R(x_1,\ldots,x_k)$, where $R=\vec 0^*\cdot L$ for some regular
language $L\subseteq (\Sigma^k)^*$. We interpret $R$ as $\eval R
\subseteq \N^k$, and the additional leading zeros we require ensure
that $R=\eval{\eval{R}}^{-1}$. Subsequently, we call a language
$L\subseteq (\Sigma^k)^*$ \emph{zero closed} if $L = \vec 0^* \cdot
L$. Given a formula $\Phi(x_1,\ldots,x_n)$ of generalised Sem{\"e}nov
arithmetic, we define $\eval \Phi \subseteq \N^d$ as the set of all
satisfying assignments of $\Phi$.

The size of an atomic formula $R(x_1,\ldots,x_k)$ is defined as the
number of states of the canonical minimal DFA defining $R$. For all
other atomic formulas $\varphi$, we define their sizes $\abs \varphi$
as the number of symbols required to write down $\varphi$, assuming
binary encoding of numbers. The size $\abs \Phi$ of an arbitrary
existential formula $\Phi$ of generalised Sem{\"e}nov arithmetic is
the sum of the sizes of all atomic formulas of $\Phi$.

The full first-order theory of generalised Sem{\"e}nov arithmetic is
known to be undecidable~\cite{Point10}. This follows from the
undecidability of $\langle \N,0,1,+,2^{(\cdot)}, V_2 \rangle$, where
$V_2$ is the binary predicate such that $V_2(x,y)$ holds if and only
if $x$ is the largest power of $2$ dividing $y$ without
remainder. Note that $V_2$ can be defined in terms of a regular
language, cf.~\cite{BHMV94}. The central result of this paper is the
following:
\begin{theorem}\label{thm:main}
  The existential fragment of generalised Sem{\"e}nov arithmetic is
  decidable in EXPSPACE.
\end{theorem}

\subsection{Affine vector addition systems with states}
A technical tool for our decidability results is a tailor-made class
of \emph{labelled affine vector addition systems with states
  (\laVASS)}. Formally, an \laVASS is a tuple $V=\langle Q,d,
\Sigma,\Delta,\lambda,q_0,F,\Phi\rangle$, where
\begin{itemize}
\item $Q$ is a finite set of \emph{control states},
\item $d\ge 0$ is the \emph{dimension of $V$},
\item $\Sigma$ is a \emph{finite alphabet},
\item $\Delta \subseteq Q\times \mathcal P(\Sigma) \times Q$ is a
  finite set of \emph{transitions},
\item $\lambda\colon \Delta \to \Ops^d$ is the \emph{update function},
  where $\Ops \subseteq \Z[x]$ is the set of all affine functions over
  a single variable,
\item $q_0\in Q$ is the \emph{initial control state},
\item $F\subseteq Q$ is the set of \emph{final control states}, and
\item $\Phi$ is a a quantifier-free formula of Presburger arithmetic
  $\Phi(x_1,\ldots,x_d)$ that specifies a finite set $\eval
  \Phi\subseteq \N^d$ of \emph{final counter values}.
\end{itemize}
Note that when $d=0$ then $V$ is essentially a non-deterministic
finite automaton.

The set of \emph{configurations} of $V$ is $C(V)\defeq Q\times
\N^d$. The \emph{initial configuration} of $V$ is
$c_0=(q_0,0,\ldots,0)$, and the set of \emph{final configurations}
is \[C_f(V) \defeq \left \{ (q_f,\vec v) : q_f\in F,\vec v\in \eval\Phi
\right \}\,.\]
%TODO: Maybe pick some other notation than v in C_f
%
%% Given a quantifier-free formula of Presburger arithmetic $\Phi$, we
%% denote by $\abs \Phi$ the number of symbols required to write down
%% $\Phi$.
For an update function $\lambda\colon \Delta \to \Ops^d$,
we define
\[
\norm \lambda \defeq \max \{ \abs a +\abs b :
\lambda(t)=(f_1,\ldots,f_d), f_i = a x + b, 1\le i\le d, t\in T\}\,.
\]
We define the size $\abs V$ of an \laVASS $V=\langle Q,d,
\Sigma,\Delta,\lambda,q_0,F,S_f\rangle$ as
\[
\abs V \defeq \abs Q + \abs \Delta \cdot (d+1) \cdot \log (\norm \lambda + 1) +
\abs \Phi \,.
\]

An \laVASS induces an (infinite) labelled directed \emph{configuration
  graph} ${G=(C(V),\xrightarrow{}})$, where ${\xrightarrow{}}
\subseteq C(V)\times \Sigma \times C(V)$ such that $c\xrightarrow{a}
c'$ if and only if
\begin{itemize}
\item $c=(q,m_1,\ldots,m_d)$ and $c'=(q',m_1',\ldots,m_d')$,
\item there is $t=(q,A,q')\in \Delta$ such that
  \begin{itemize}
  \item $a\in A$,
  \item $\lambda(t)=(f_1,\ldots,f_d)$, and
  \item $m_i'=f_i(m_i)$ for all $1\le i\le d$.
  \end{itemize}
\end{itemize}
We lift the definition of $\xrightarrow{}$ to words $w=a_1\cdots a_n
\in \Sigma^*$ in the natural way, and thus write $c\xrightarrow{w}c'$
whenever $c \xrightarrow{a_1} c_1 \xrightarrow{a_2} \cdots c_{n-1}
\xrightarrow{a_n} c'$ for some $c_1,\ldots,c_{n-1} \in C$. The
\emph{language} $L(V) \subseteq \Sigma^*$ of $V$ is defined as
\[
L(V) \defeq \{ w \in \Sigma^* : c_0 \xrightarrow{w} c_f, c_f \in C_f(V)
\}\,.
\]
If we are interested in runs of \laVASS, we write $\pi = c_1\xrightarrow{t_1} c_2\xrightarrow{t_2}
\cdots \xrightarrow{t_{n-2}}c_{n-1} \xrightarrow{t_{n-1}} c_n$ to emphasise the sequence of configurations and 
transitions taken. For $1\le i\le j\le n$, we denote by $\pi[i,j]$ the
subsequence $c_i\xrightarrow{t_i} c_{i+1}\xrightarrow{t_{i+1}} \cdots \xrightarrow{t_{j-1}}c_j$. We denote by
$\val(\pi,m_i)$ the value $m_i$ of the $i$-th counter in the last configuration of $\pi$.

The \emph{emptiness problem} for an \laVASS is to decide whether
$L(V)\neq \emptyset$. Affine VASS are a powerful class of infinite
state systems, and even in the presence of only two counters and
$\Phi(x_1,x_2) \equiv x_1 = x_2$, the emptiness problem is
undecidable~\cite{Reichert2015}.  In \Cref{sec:emptiness-rlavass}, we
identify a syntactic fragment of \laVASS for which emptiness can be
decided in EXPSPACE.

\subsubsection*{Closure properties of languages of \laVASS}

We briefly discuss closure properties of \laVASS and show that they
are closed under union and intersection, and restricted kinds of
homomorphisms and inverse homomorphisms, using essentially the
standard constructions known for finite-state automata. Let
$V_i=\langle Q_i,d_i,
\Sigma,\Delta_i,\lambda_i,q_0^{(i)},F_i,\Phi_i\rangle$, $i\in
\{1,2\}$, be two \laVASS.

\begin{proposition}\label{prop:union-intersetion-closure}
  The languages of \laVASS are closed under union and intersection.
  Moreover, for $V$ such that $L(V) = L(V_1) \cap L(V_2)$, we have
  $\abs{V} \le \abs{V_1} \cdot \abs{V_2}$.
\end{proposition}
\begin{proof}
  This result can be obtained by generalising the standard
  constructions known from non-deterministic finite-state automata.
  The set of control states of the \laVASS $V$ accepting the
  intersection of \laVASS $V_1$ and $V_2$ is $Q_1\times Q_2$. The
  dimension of $V$ is the sum of the dimensions of $V_1$ and $V_2$,
  and the counters of $V_1$ and $V_2$ get independently simulated in
  the counters of $V$. Upon reading an alphabet symbol $a$, the
  \laVASS $V$ then simultaneously simulates the respective transitions
  of $V_1$ and $V_2$ for $a$; further details are relegated to
  \Cref{app:intersection}.
\end{proof}

Note that since \laVASS languages contain regular languages,
\Cref{prop:union-intersetion-closure} in particular enables us to
intersect \laVASS languages with regular languages.

Let $\Sigma,\Gamma$ be two finite alphabets. Recall that a
homomorphism $h\colon \Gamma^* \to \Sigma^*$ is fully defined by
specifying $h(a)$ for all $a\in \Gamma$. We call $h$ a
\emph{projection} if $\abs{h(a)}=1$ for all $a\in \Gamma$.
\begin{proposition}\label{prop:projection-closure}
  The languages of \laVASS are closed under projections and inverses
  of projections.
\end{proposition}
\begin{proof}
  Let $h\colon \Gamma^* \to \Sigma^*$ be a projection. Given an
  \laVASS $V=\langle Q,d, \Sigma,\Delta,\lambda,q_0,F,S_f\rangle$, to
  obtain closure under projections replace any $t=(q,A,q')\in \Delta$
  with $t'=(q,h(A),q')$, and set $\lambda(t')\defeq \lambda(t)$. To
  obtain closure under inverse projections, replace any $t=(q,A,q')\in
  \Delta$ with $t'=(q,h^{-1}(A),q')$ and set $\lambda(t')\defeq
  \lambda(t)$.
\end{proof}

\section{Reducing Sem{\"e}nov arithmetic to restricted \laVASS}\label{sec:semenov-to-lavass}

Let $\Sigma=\{0,1\}$. In this section, we show how given a
quantifier-free formula $\Phi(x_1,\ldots,x_d)$ of Sem{\"e}nov
arithmetic, we can construct an \laVASS $V$ over the alphabet
$\Sigma_d\defeq \{0,1\}^d$ such that $\eval{L(V)}= \{\vec
x\in\N^d:\Phi(\vec x)\}$. We will subsequently observe that the
resulting \laVASS enjoy strong structural restrictions, giving rise to
the fragment of restricted \laVASS that we then formally define. For
our purposes, it will be sufficient to primarily focus on formulas
$\Phi$ of Sem{\"e}nov arithmetic which are conjunctions of atomic
formulas.
\begin{lemma}\label{lem:conjunctive-lavass}
  Consider a positive conjunctive formula of Sem{\"e}nov arithmetic
  \[
  \Phi(\vec x) \equiv A\cdot \vec x = \vec b \wedge
  \bigwedge_{i\in I} x_i = 2^{y_i},
  \]
  where $A\in \Z^{m\times n}$, $\vec b\in \Z^m$, $I$ is
  a finite index set, and $x_i$ and $y_i$
  are variables from $\vec x$. There is an \laVASS $V$ of dimension
  $2 \abs I$ and of size $(\norm A_{1,\infty} + \norm{\vec b}_\infty +
  2)^{O(m+\abs I)}$ such that $L(V)=\eval \Phi$.
\end{lemma}

We derive this lemma in two parts. First, it is well known that the
sets of solutions of the systems of linear equations $A\cdot \vec x
= \vec b$ can be represented by a regular language and are hence definable via
an \laVASS.
%\cite{WB00}, 
\begin{lemma}[\cite{WB00}, see also {\cite[Eqn.~(1)]{GuepinHW19}} ]\label{lem:equations-lavass}
  Given a system of equations $\Phi \equiv A\cdot \vec x = \vec b$
  with $A \in \Z^{m\times d}$ and $\vec b\in \Z^m$, there is a DFA $V$
  with at most $2^m \cdot \max\{ \norm A_{1,\infty}, \norm{\vec
  b}_\infty \}^m$ states such that $L(V)$ is zero-closed and
  $\eval{L(V)}=\eval \Phi$.
\end{lemma}

The crucial part, which requires the power of \laVASS, are exponential
equations $x=2^y$. An \laVASS $V$ with two counters and
$\eval{L(V)}=\eval{x=2^y}$ is depicted
in \Cref{fig:exponential-gadget}. Control-states are depicted as
circles and transitions as arrows between them. The vector before the
colon is the alphabet symbol read. For instance, the transition from
$q_0$ to $q_1$ reads the alphabet symbol $(1,0)\in \{0,1\}^2$. After a
colon, we display the counter operations when reading the alphabet
symbol, the operation on the first counter is displayed on the top and
the operation on the second counter on the bottom. Here and
subsequently, for presentational convenience, $\wID$ is the identity
$x\mapsto x$, and $\wTT$ and $\wTTPO$ are the functions $x \mapsto 2
x$ and $x \mapsto 2 x + 1$, respectively.  Thus, the transition from
$q_0$ to $q_1$ applies the identity function on the first counter, and
the function $x\mapsto 2x$ on the second counter.
\begin{figure}[!t]
\centering
\begin{tikzpicture}[->,shorten >=1pt,auto,node distance=4cm,semithick]

\tikzstyle{every state}=[draw=black,text=black,minimum size=25pt]

\node[initial by arrow,state,initial text=] (1) {$q_0$};
\node[state, accepting] (2) [right of=1] {$q_1$};

\path 
    (1) edge node [align=center]{\small{$\left[\begin{matrix} 1\\0 \end{matrix}\right] : \begin{array}{c}
         \wID  \\
         \wTT
    \end{array}$}}  (2)
    (1) edge [loop above] node {\small{$\left[\begin{matrix} 0\\0 \end{matrix}\right] : \begin{array}{c}
         \wID  \\
         \wTT
    \end{array}$}} (1)
    (2) edge [loop above] node {\small{$\left[\begin{matrix} 0\\b \end{matrix}\right] : \begin{array}{c}
         \wPP  \\
         \wTT + b
    \end{array}$}} (2);
\end{tikzpicture}
  \caption{Gadget with two counters for exponential equations $x=2^y$,
  where $b\in\{0,1\}$.}  \label{fig:exponential-gadget}
\end{figure}
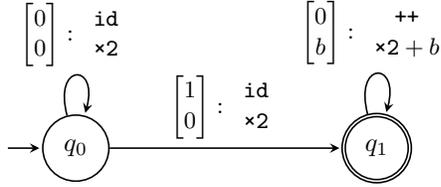

The idea behind the gadget in \Cref{fig:exponential-gadget} is as
follows. For an example, suppose that $y=5$, then $x=32$, and in
binary the sequence of digits of $x$ and $y$ looks as follows:
\[
\left[
\begin{smallmatrix}
  x \\ y
\end{smallmatrix}
\right] = \left[
\begin{smallmatrix}
  1 \\ 0
\end{smallmatrix}
\right]
\left[
\begin{smallmatrix}
  0 \\ 0
\end{smallmatrix}
\right]
\left[
\begin{smallmatrix}
  0 \\ 0
\end{smallmatrix}
\right]
  \cdots
\left[
\begin{smallmatrix}
  0 \\ 0
\end{smallmatrix}
\right]
\left[
\begin{smallmatrix}
  0 \\ 1
\end{smallmatrix}
\right]
\left[
\begin{smallmatrix}
  0 \\ 0
\end{smallmatrix}
\right]
\left[
\begin{smallmatrix}
  0 \\ 1
\end{smallmatrix}
\right]
\]
Since $x=2^y$, we have that $x\in 0^*10^*$, and the number of tailing
zeros of $x$ is equal to the value of $y$. Thus, once a $1$ in the
binary representation of $x$ has been read, the first counter in the
gadget of \Cref{fig:exponential-gadget} keeps incrementing and counts
the number of tailing zeros of $x$. At the same time, the second
counter in the gadget of \Cref{fig:exponential-gadget} keeps the value
0 until it reads the first $1$ of the binary expansion of $y$, since
$2\cdot 0=0$. It then computes the value of $y$ in binary on the
second counter by multiplying the value of the second counter by 2
when reading a zero, and multiplying by two and adding one when
reading a one. The \laVASS in \Cref{fig:exponential-gadget} only
accepts when the first and the second counter have the same value,
i.e., when the number of tailing zeros of the binary expansion of $x$
equals the value of $y$, as required.
\begin{lemma}\label{lem:exponential-lavass}
  There is a fixed \laVASS $V$ of dimension two such that $L(V)$ is
  zero closed and $\eval{L(V)}=\eval{x = 2^y}$.
\end{lemma}

\Cref{lem:conjunctive-lavass} is now an easy consequence of
\Cref{lem:equations-lavass,lem:exponential-lavass} together with the
closure of \laVASS languages under intersection
(\Cref{prop:union-intersetion-closure}) and inverse homomorphisms
(\Cref{prop:projection-closure}).

A closer look at the gadget constructed in
\Cref{fig:exponential-gadget} reveals a number of important structural
properties:
\begin{enumerate}[(i)]
\item all operations performed on the first counter are either the
  identity map $\wID$ or increments $\wPP$;
\item all operations performed on the second counter are affine
  updates $\wTT$ and $\wTTPO$;
\item once the first counter gets incremented on a run, it gets
  incremented at every subsequent transition; and
\item only counter configurations in which the value of the first
  counter equals the value of the second counter are accepted.
\end{enumerate}
Those properties are crucial to obtain decidability of (generalised)
existential Sem{\"e}nov arithmetic. 
\begin{definition}
  An \laVASS is \emph{restricted} if it has an even number of $2d$
  counters called $x_i,y_i$, $1\le i\le d$, such that every counter
  pair $(x_i,y_i)$ adheres to the above Properties~(i)--(iv), and the
  set of final counter values is defined by
  $\Phi \equiv \bigwedge_{1\le i\le d} x_i=y_i$.
\end{definition}
For convenience, when referring to the counters in a pair, we
subsequently refer to the first counter as its \emph{$x$-counter} and to
the second counter as its \emph{$y$-counter}. We will usually write $m$
for the value of the $x$-counter and $n$ for the value of the
$y$-counter. The following is immediate from the construction in
\Cref{prop:union-intersetion-closure}:
\begin{proposition}
The languages of restricted \laVASS are closed under union,
intersection, projection and inverse projections.
\end{proposition}
Finally, subsequently, for technical convenience, we assume that for a
restricted \laVASS, we have $\abs Q\ge 2$. This is with no loss of
generality, since if $\abs Q=1$ then deciding emptiness is trivial
(the only control state is accepting if and only if the restricted
\laVASS has non-empty language).

The next section will be devoted to the proof of the main result of
this paper on restricted \laVASS:
\begin{proposition}\label{prop:main}
  Language emptiness of a restricted \laVASS $V$ with $2d$ counters is
  decidable in $\mathrm{NSPACE}( \abs V \cdot 2^{O(d)})$.
\end{proposition}

Let us close this section with arguing how \Cref{thm:main} follows
from \Cref{prop:main}. Given a formulas $\Phi$ of generalised
Sem{\"e}nov arithmetic, we can in space $2^{O(|\Phi|)}$ construct the
disjunctive normal form of $\Phi$. Every disjunct can be assumed
to be of the form
\[
A\cdot \vec x = \vec b \wedge \bigwedge_{i\in I} x_i = 2^{y_i} \wedge
\bigwedge_{j\in J} R_j(\vec x),
\]
where the $R_j$ are predicates over regular languages. By
\Cref{lem:conjunctive-lavass}, there is a restricted \laVASS for $\Phi$ of
dimension $2 \abs I$ with a number of states bounded by $(\norm
A_{1,\infty} + \norm{\vec b}_\infty + 2)^{O(m+\abs
I)}=2^{p(\abs \Phi)}$ for some polynomial $p$ and whose language
represents the set of solutions to $A\cdot
\vec x = \vec b \wedge \bigwedge_{i\in I} x_i = 2^{y_i}$. Intersecting
with the DFA for the $R_j$ results in a restricted \laVASS $V$ with
$2\abs I=O(\abs \Phi)$ counters such that $\abs V = 2^{p(\abs \Phi)}$
for some polynomial $p$. By \Cref{prop:main}, it follows that
emptiness of $V$ is decidable in non-deterministic space exponential in
$p(\abs \Phi)$. We conclude the argument by recalling that NEXPSPACE=EXPSPACE by Savitch's theorem.

\section{Emptiness certificate for restricted \laVASS}\label{sec:emptiness-rlavass}

We now show that language emptiness for restricted \laVASS is
decidable in exponential space. Clearly, this problem reduces to
deciding whether a given restricted \laVASS has an accepting run, but
witnessing runs may be of non-elementary length. To overcome this
problem, we define an abstraction for configurations of
restricted \laVASS. Abstract configurations store residue classes of
counter values, as well as some further information that is required
to witnesses the existence of concrete accepting runs. Before giving
the formal definition, we provide some high level intuition that leads
to our definition of abstract configurations. Next, we introduce
reachability certificates, which are abstract runs with certain
further properties. We argue that the existence of \emph{witnessing
certificates}, which are special kinds of reachability certificates
witnessing that the language of an \laVASS is non-empty, are decidable
in EXPSPACE. The last two sections then establish that witnessing
certificates actually witness non-emptiness of restricted \laVASS.

\subsection{Key observations}

Given a \emph{restricted} \laVASS $V$ in dimension $d$, assuming that
$L(V)\neq \emptyset$, there is a run $\pi$ from an initial
configuration $c$ to a final configuration $c'$. With no loss of
generality, throughout this section, we assume that $\val(c',x_i) \ge
\val(c',x_{i+1})>0$ for all $1\le i<d$. In particular, this implies
that every counter gets incremented at least once along a path
witnessing non-emptiness.

Our first observation is that if along $\pi$ a counter $y_i$ achieves
the first time a non-zero value by taking a $\wTTPO$ labeled
transition, the length of the remaining segment of $\pi$ is bounded by
$O(\log(m_i+1))$, where $m_i$ is the value of counter $x_i$ before the
transition is taken. The reason is that, once $y_i$ has non-zero
value, its value at least doubles when a transition is taken. Hence if
$\pi$ is ``long'' then along $\pi$ there will be loops incrementing a
counter $x_i$ before the corresponding $y_i$ achieves non-zero value.

In the latter scenario, we may actually, subject to some bookkeeping,
discard concrete values of $x_i$ and $y_i$ and only store their
residue classes modulo $\ell_i$, where $\ell_i$ is the length of the
first loop incrementing $x_i$ along $\pi$. In particular, if we are
given a non-accepting run $\pi'$ such that $\val(\pi',x_i) \equiv
\val(\pi',y_i) \bmod \ell_i$ and $\val(\pi',x_i)<\val(\pi',y_i)$ then
$\pi'$ can be turned into a run $\pi''$ where $\val(\pi'',x_i) =
\val(\pi'',y_i)$ by iterating the loop of length $\ell_i$.

There are, however, some further subtleties that need to be taken care
of. Consider the segment $\pi'$ of $\pi$ between the first transition
labeled by $\wPP$ on $x_i$ and the first transition labeled by $\wPP$ on
$x_{i+1}$. If $\pi'$ contains no loop then we are in a situation where
the first loop incrementing $x_i$ is also the first loop incrementing
$x_{i+1}$. This means that the values of $x_i$ and $x_{i+1}$ get
paired together, and hence, for an accepting run, also the values of
$y_i$ and $y_{i+1}$ are paired together. In our approach, we deal with
such circumstances by introducing so-called \emph{$y$-constraints}. A
$y$-constraint of the form $y_i-y_{i+1}=\delta_i$ for some constant
$\delta_i \in \N$ asserts that the counters $y_{i+1}$ and $y_i$ must
eventually have constant difference $\delta_i$ along a run.

Otherwise, if $\pi'$ above contains a loop, the difference between the
values of $x_i$ and $x_{i+1}$ is not necessarily constant, but
lower-bounded by the length $\delta_i$ of the loop-free sub path of
$\pi'$. Thus, in an accepting run, the difference between $y_i$ and
$y_{i+1}$ must also be at least $\delta_i$, which is asserted by a
$y$-constraint of the form $y_i-y_{i+1} \ge \delta_i$.

\subsection{An abstraction for restricted \laVASS}

Our decision procedure for emptiness of restricted \laVASS is based on
reducing this problem to a reachability problem in a
carefully designed finite-state abstraction of the state-space of
\laVASS. Throughout this section, let $V=\langle Q,2d,
\Sigma,\Delta,\lambda,q_0,F,\Phi\rangle$ be a restricted \laVASS.  We
first define the state space of the abstracted \laVASS.
\begin{definition}
  An abstract configuration is a tuple
  \begin{multline*}
  \alpha=(q,
  m_1,n_1,\dots,m_{d},n_{d},u_1,u_2,\dots,u_{d-1},\ell_1,\dots,\ell_d)\\
  \in Q\times (\N\cup\{\bot\})^{2d} \times \N^{d-1}\times(\N \cup \{ \top \})^{d},
  \end{multline*}
  such that $m_i,n_i \in [0,2dM_i] \cup \{ \bot \}$
  and $u_i \in [0,U_i]$ and $\ell_i\in[0,M_i-1] \cup \{\top\}$,
  where
  \begin{itemize}
	\item $M_i \defeq \floor{\abs Q^{((1/8)\cdot 32^{i-1} + 1)}}$; and
	\item $U_i \defeq \abs Q^{(32^{i-1}+4)}$.
  \end{itemize}
\end{definition}
The idea is that $m_i,n_i$ store the residue classes modulo $\ell_i$
of the counter pair $x_i,y_i$ respectively, where the value $\top$ for
$\ell_i$ acts as an indicator that we are storing actual values and
not residue classes. The value $\bot$ for some $x_i$ or $y_i$
indicates that the counter has not yet been initialised. If for an
update function $f$, $f=\wPP$ or $f=\wTTPO$ then $f(\bot)\defeq
1$; otherwise $f(\bot) \defeq \bot$, and we stipulate that $\bot \bmod
n = \bot$. The value of $u_i$ in an abstract configuration carries the
current difference between the value of the counters $y_{i}$ and
$y_{i+1}$. This difference is potentially unbounded, however for our
purposes it suffices to only store its value if it is less than $U_i$,
and to indicate the fact that it is at least $U_i$ by the value
$u_i=U_i$.
%This difference must not become negative since $u_1,\dots,u_{d-1}$ are
%all natural numbers.

We denote the (finite) set of all abstract configurations of $V$ by
$A(V)$. Let us now define a transition relation
${\xrightarrow{\cdot}}\subseteq A(V) \times \Delta \times A(V)$ such
that $\alpha\xrightarrow{t}\alpha'$, $t=(q,a,q')\in \Delta$ if and
only if:
\begin{itemize}
    \item $\alpha=(q,m_1,n_1,\dots,u_{d-1},\ell_1,\dots,\ell_d)$ and $\alpha'=(q',m_1',n_1',\dots,u_{d-1}',\ell_1,\dots,\ell_d)$;
    \item $\lambda(t)=(f_{x_1},f_{y_1},\dots,f_{x_d},f_{y_d})$;
    \item $f_{x_i} = \wPP$ for all $i$ such that $m_i\neq \bot$;
    \item if $\ell_i\neq\top$, $m_i' = f_{x_i}(m_i)\bmod \ell_i$ and $n_i'=f_{y_i}(n_i)\bmod \ell_i$;
    \item if $\ell_i=\top$, $m_i' = f_{x_i}(m_i)$ and $n_i'=f_{y_i}(n_i)$; and
    \item for all $i\in\{1,\dots,d-1\}$,
      $$u_i' =
    \begin{cases*}
      min(2u_i+1,U_i) & if $f_{y_{i}}=\wTTPO, f_{y_{i+1}}=\wTT$ \\
      min(2u_i,U_i) & if $f_{y_{i}}=f_{y_{i+1}}$ \\
      min(2u_i-1,U_i) & if $f_{y_{i}}=\wTT, f_{y_{i+1}}=\wTTPO$. \\
    \end{cases*}
    $$
\end{itemize}
Assuming that the value of $y_i$ is at least the value of $y_{i+1}$,
which we will always ensure, the definition of how to update $u_i$
ensures that it exactly stores the difference $y_i-y_{i+1}$ unless the
difference becomes too large, in which case it is levelled off at
$U_i$.

An abstract configuration path is a sequence of abstract
configurations and transitions of the form
$R=\alpha_1\xrightarrow{t_1} \alpha_2 \xrightarrow{t_2} \cdots
\xrightarrow{t_{n-1}} \alpha_n$.

Given two consecutive $y$-counters $y_i,y_{i+1}$ and $\delta_i\in \N$,
we say that $y_{i}-y_{i+1}=\delta_i$ and $y_{i}-y_{i+1}\geq \delta_i$
are \emph{$y$-constraints}. Let $Y$ be a set of $y$-constraints, an
abstract configuration
$\alpha=(q,m_1,n_1,\dots,u_1,\dots,u_{d-1},\ell_1,\dots,\ell_d)$
\emph{respects $Y$} whenever
\begin{itemize}
\item
  $u_i\geq \delta_i$
  for all constraints of type $y_{i}-y_{i+1} \geq \delta_i$ in $Y$,
\item and $u_i<U_i$ and $u_i=\delta_i$ for all constraints
  $y_{i}-y_{i+1} = \delta_i$ in $Y$.
\end{itemize}
We say that $\alpha_f$ is a \emph{final abstract configuration
  respecting $Y$} whenever $q\in F$, $m_i=n_i$ for all $1\leq i\leq
d$, and $\alpha_f$ respects $Y$.

\subsection{Witnessing certificates}

While any concrete accepting run of an \laVASS gives rise to an
abstract configuration path ending in an accepting abstract
configuration, the converse does not hold. This motivates the
introduction of reachability and witnessing certificates, which are
special abstract configuration paths that carry further information
that eventually enables us to derive from a witnessing certificate a
concrete accepting run of an \laVASS.

A \emph{reachability certificate} is a tuple $(R,X,Y,L)$ such that
$R=\alpha_1\xrightarrow{t_1} \alpha_2 \xrightarrow{t_2} \cdots
\xrightarrow{t_{n-1}} \alpha_n$ is an abstract configuration path, and
$X,Y,L\colon \{1,\ldots,d\} \to \{1,\ldots,n\}$. Here, $X(i)$ and
$Y(i)$ indicate the position where the $x_i$-counter and $y_i$-counter
obtain a value different from $\bot$ for the first time. Moreover,
$L(i)$ is the position where a loop of length $\ell_i$ can be found.
Formally, $(R,X,Y,L)$ is required to have the following properties:
\begin{enumerate}[(a)]
\item $\alpha_1 = (q_0,\bot, \ldots, \bot,0,\ldots,0,
  \ell_1,\ldots,\ell_d)$ and if $\ell_i=\top$ then $\ell_j=\top$ for
  all $i<j\le d$;
\item $\lambda(x_i,t_{X(i)-1})=\wPP$ and $\lambda(y_i,t_{Y(i)-1})=\wTTPO$ for all $1\le i\le d$;
\item $\lambda(x_i,t_{j})=\wID$ for all $1\le j < X(i)-1$;
\item $\lambda(y_i,t_{j})=\wTT$ for all $1\le j < Y(i)-1$;
\item $X,Y,L$ are monotonic;
\item for all $1\le i\le d$, if $\ell_i\neq \top$ then
  \begin{itemize}
  \item $X(i)\le L(i)< Y(i)$; and
  \item there is a simple $\alpha_{L(i)}$-loop $\alpha_{L(i)}
    \xrightarrow{t_1'} \alpha'_2 \xrightarrow{t_2'} \cdots
    \xrightarrow{t_{\ell_i-1}} \alpha'_{\ell_i-1}
    \xrightarrow{t_{\ell_i}'} \alpha_{L(i)}$ of length $\ell_i$.
  \end{itemize}
\end{enumerate}
Those conditions can be interpreted as follows. Condition~(a) asserts
that the certificate starts in an initial abstract configuration. We
require that $\top$ monotonically propagates since since the absence of
a loop for counter $x_i$ implies that the remainder of a path is
short, hence we can afford to subsequently store actual counter
values and not residue classes. Conditions~(b), (c) and~(d) assert
that $X(i)$ and $Y(i)$ are the first position where the counters
$x_i,y_i$ hold value different from $\bot$. Condition~(e) states that
the counters $x_{i+1}$, $y_{i+1}$ do not carry a value different from
$\bot$ before the counters $x_{i}$ and $y_{i}$,
respectively. Condition~(f) implies that, if $\ell_i\neq \top$ then
between the first update for counter $x_i$ and the first update for
counter $y_i$ there is a position $L(i)$ where we can find a loop in
the abstract configurations of length $\ell_i$. Notice that if
$x_j=\bot$ or $y_j=\bot$ in $\alpha_{L(i)}$ then $x_j$ and $y_j$
remain to hold $\bot$ along this loop, i.e., this loop does not update
counters that have not been initialised already.

Given $R$, the set of $y$-constraints induced by $R$ is the smallest
set containing
\begin{itemize}
\item $y_{i}-y_{i+1}\geq \delta_i$, where $\delta_i \defeq
  X(i+1)-X(i)$ if there is a $j$ such that $X(i)\le L(j) < X(i+1)$;
  and
\item otherwise $y_{i}-y_{i+1} = \delta_i$, where $\delta_i \defeq
  X(i+1)-X(i)$,
\end{itemize}
for all $1\le i< d$ such that $\ell_i\neq\top$.

We introduce some further notation. Given a reachability certificate
$R$, we denote by $\pi(R)$ the run corresponding to $R$ in the
configuration graph of $V$, with the initial configuration
$(q_0,0,0,\ldots,0,0)$. Given indices $1\le i\le j\le n$, we denote by $R[i,j]$ the
segment $\alpha_i \xrightarrow{t_i} \alpha_{i+1} \cdots
\xrightarrow{t_{j-1}} \alpha_j$ of $R$, and by $R[i]\defeq \alpha_i$.
We say that $R$ is a \emph{witnessing certificate} if, for $a\le d$
being the largest index such that $\ell_a \neq \top$:
\begin{itemize}
\item $R[1,Y(a)]$ is a simple path and $n-Y(a)\le 2dM_{d+1}$;
\item $\alpha_n$ is a final abstract configuration respecting the set
  of induced $y$-constraints; and
\item $\val(\pi(R),x_a) \le \val(\pi(R),y_a)$.
\end{itemize}
Sometimes we will speak of witnessing certificates \emph{restricted}
to a set of counters. By that we mean a witnessing certificates where
the relevant Conditions~(a)--(f) are only required for that set of
counters.

Now we are ready to provide a proof for \Cref{prop:main}, that stated that language emptiness for restricted \laVASS can be decided in $\mathrm{NSPACE}( \abs V \cdot 2^{O(d)})$.
\begin{proof}[Proof of \Cref{prop:main}]
  Clearly, an abstract configuration can be stored in space $\abs V \cdot 2^{O(d)}$. An NEXPSPACE algorithm can hence
  non-deterministically choose an initial configuration and
  non-deterministically verify that it leads to a final abstract
  configuration along a path that is a witnessing certificate. To this
  end, the algorithm computes the set of induced $y$-constraints
  on-the-fly while guessing the reachability certificate, and verifies
  them in the last configuration. Note that the $y$-constraints can be stored in space $\abs V \cdot 2^{O(d)}$. Finally, the requirement
  $\val(\pi(R),x_a) \le \val(\pi(R),y_a)$ can also be verified in
  exponential space since we require that $R[1,Y(a)]$ is a simple path and $n-Y(a)\le 2M_{d+1}$.
\end{proof}
In the next section we argue the correctness of our algorithm by proving the following theorem:
\begin{theorem}\label{theorem:certificate iff non-empty language}
The language of a restricted \laVASS $V$ is non-empty if and only if there exists a witnessing certificate for $V$.
\end{theorem}

\section{Correctness proof of the certificate}
In this section we prove \Cref{theorem:certificate iff non-empty
language}. The proof is split into the two directions. In
\Cref{ssec:certificate-to-reachability} below, we show that the existence of a witnessing
certificate for an \laVASS implies that the language of the \laVASS is
non-empty. The converse direction is then shown
in \Cref{ssec:non-empty-to-certificate}.

\subsection{Witnessing certificates imply language non-emptiness}\label{ssec:certificate-to-reachability}

This section proves the following proposition.
\begin{proposition}\label{lem:abstract-to-concrete}
  If there exists a witnessing certificate for a restricted \laVASS
  $V$ then $L(V)\neq \emptyset$.
\end{proposition}

The idea behind the proof of \Cref{lem:abstract-to-concrete} is that
we obtain from a witnessing certificate $(R,X,Y,L)$ of an \laVASS $V$
a sequence of runs of $V$ such that the final run in that sequence is
an accepting run of $V$. Initially, we obtain a run that ends in a
configuration where the counters are in a congruence relation. We then
carefully pump the simple loops pointed to by $L$, beginning from the
last counter working towards the first.

To formally prove \Cref{lem:abstract-to-concrete}, let $(R,X,Y,L)$ be
a witnessing certificate, and let $\pi(R)$ be the run in the
configuration graph of $V$ induced by $R$. Let $a\le d$ be maximal
such that $\ell_a\neq \top$. We now define a sequence of runs
$\pi_0,\ldots,\pi_a$ such that the following invariant holds. In the final
configuration of $\pi_i$,
\begin{enumerate}[(i)]
\item $m_j \le n_j$ and $m_j\equiv n_j \bmod \ell_j$ for the $j$-th
  counter pair, $1\le j \le a-i$; and
\item $m_j = n_j$ for the $j$-th counter pair, $a-i<j\le d$.
\end{enumerate}
It is clear that $\pi_a$ then witnesses $L(V)\neq \emptyset$. We
proceed by induction on $i$.

\emph{Base case $i=0$:} Let $\pi_0=\pi(R)$. Since $R$
is a witnessing certificate, $val(\pi(R),x_a)\leq val(\pi(R),y_a)$, and hence
$m_a \le n_a$ in the last configuration of $\pi_0$.
%Since $R$ is a
%witnessing certificate, we also have $m_{a+1}=n_{a+1}, \ldots,
%m_d=n_d$ by definition.
Moreover, $R$ respects the set of induced $y$-constraints. Hence
$n_{a-1} - n_a \ge \delta_{a-1}$, where $\delta_{a-1}$ is the length
of the path from $R[X(a-1)]$ to $R[X(a)]$. Hence $n_{a-1} - n_a \ge
m_{a-1} - m_a$ and thus $m_{a-1} \le n_{a-1}$. Iterating this argument
for the remaining counters, we get that~(i) of the invariant is
fulfilled for $\pi_0$; (ii) trivially holds since $R$ ends in an
accepting abstract configuration.

\emph{Induction step  $i>0$:} Let 
$\pi_{i-1}$ be the path that exists by the induction hypothesis. If
$m_{a-i}=n_{a-i}$ in the last configuration of $\pi_{i-1}$ then we are
done and take $\pi_i=\pi_{i-1}$; otherwise $m_{a-i} < n_{a-i}$ and
$m_{a-i} \equiv n_{a-i} \bmod
\ell_{a-i}$. Hence, there is some $k\in \N$ such that $n_{a-i}= k
\cdot \ell_{a-i}$. Since $\ell_i\neq\top$, let
$\beta\defeq\alpha_{L(a-i)} \xrightarrow{t_1} \alpha_2
\xrightarrow{t_2} \cdots
\alpha_{\ell_i-1}\xrightarrow{t_{\ell_i}}\alpha_{L(a-i)}$ be the
simple $\alpha$-loop at position $L(a-i)$ that is guaranteed to exist
since $R$ is a witnessing certificate. We insert the transitions of
$\beta^k$ and the induced updated configurations into $\pi_{i-1}$ at
position $L(a-i)$. Notice that $L(a-i) < X(a-i+1)$. Otherwise, by the
definition of the induced $y$-constraints,
$y_{a-i}-y_{a-i+1}=\delta_{a-i}$ is in the set of induced
$y$-constraints, where $\delta_i= X(a-i+1)-X(a-i)$. Since the last
abstract configuration of $R$ respects the set of $y$-constraints, it
must be the case that in the last configuration of $\pi_{i-1}$,
$n_{a-i}-n_{a-i+1}= \delta_{a-i}$ and
$m_{a-i}-m_{a-i+1}=\delta_{a-i}$, so $m_{a-i}=n_{a-i}$, because after
the position $X(a-i+1)-1$ in $R$ and thus $\pi_{i-1}$, the counters
$x_{a-i},x_{a-i+1}$ get incremented simultaneously. This contradicts
our assumption that $m_{a-i}\neq n_{a-i}$. Thus, the counters
$x_{a-i+1},y_{a-i+1},\ldots,x_d,y_d$ remain unchanged by the insertion
of $\beta^k$, so (ii) and consequently~(i) continues to hold in the
last configuration of $\pi_i$ for those counters. Moreover, due to the
ordering conditions imposed on witnessing certificates, the value of
$y_{a-i}$ does not change either, and hence $m_{a-i}=n_{a-i}$ in the
last configuration of $\pi_i$. Since $\beta$ is a loop in the abstract
configuration space, we have $m_j \equiv n_j \bmod \ell_j$ for all
$1\le j< a-i$ and the values of $u_j$, for all $1\le j<a$ are
preserved.

\subsection{Reachability yields witnessing certificates}\label{ssec:non-empty-to-certificate}

We now turn towards the converse direction and show that we can obtain
a witnessing certificate from a run witnessing non-emptiness. 
\begin{proposition}\label{lem:concrete-to-abstract}
  If a restricted \laVASS $V$ admits an accepting run then
  there exists a witnessing certificate for $V$.
\end{proposition}
We begin with defining a function that turns a configuration from
$C(V)$ into an abstract configuration. This function is parameterised
by $\ell_1,\ldots,\ell_d\in \N_+ \cup \{ \top \}$:
\begin{multline*}
f_V((q,m_1,n_1,\dots,m_d,n_d),\ell_1,\dots,\ell_d)\defeq 
(q,m_1\circ \ell_1,n_1\circ \ell_1,\dots,m_d\circ \ell_d,n_d\circ \ell_d,\\
min(n_1-n_2,U_1),\dots,min(n_{d-1}-n_{d},U_{d-1}),\ell_1,\dots,\ell_d)\,.
\end{multline*}
Here, $m\circ \ell\defeq\bot$ if $m=0$; $m\circ \ell\defeq
m\bmod \ell$ if $\ell\in\N_+$; and $m\circ \ell\defeq m$ if
$\ell=\top$.  We lift the definition of $f_V$ to paths of concrete
runs $\pi$ in the natural way, and write
$f_V(\pi,\ell_1,\dots,\ell_d)$ for the resulting sequence of abstract
configurations. Let $\pi=c_1 \xrightarrow{t_1}
c_2 \cdots \xrightarrow{t_{n-1}} c_n$ be a run witnessing
$L(V)\neq\emptyset$. We show how to obtain a witnessing certificate
$R$ from $\pi$. Without loss of generality, in $c_n$ we have $m_1\ge
m_2 \ge
\ldots m_d>0$.

To this end, we show how from the accepting run $\pi$ we can
iteratively define a sequence $R_0,R_1,R_2\ldots,R_d$ of abstract runs
and identify the required $\ell_1,\ldots, \ell_d \in
\N_+ \cup \{ \top \}$ and $X,Y,L$ such that $(R_d,X,Y,L)$ is a reachability
certificate. Let $X(i)\defeq j$ such that $j$ is the first position in
$\pi$ where the value of counter $x_i$ is non-zero; analogously define
$Y(i)$ to be the first position where the value of $y_i$ is
non-zero. Clearly, $X,Y$ are monotonic and $X(i)\le Y(i)$, for all
$1\le i\le d$. Otherwise, if a counter $y_i$ gets initialised before
the counter $x_i$ in $\pi$, it must be the case that $n_i>m_i$ in
$c_n$ and therefore $c_n$ cannot be an accepting configuration.

Recall that $\pi$ has length $n$. In our proof, the subsequent
technical lemma will allow us to conclude that, if for a counter pair
$x_i, y_i$ the $y_i$ counter gets updated shortly after the $x_i$
counter then the run will end shortly after and counter pairs
$x_j,y_j$ for $j\ge i$ will consequently have small values.
\begin{lemma}\label{lm:smallpath}
If $Y(i)-X(i)\le dM_i$ for some $1\le i\le d$ then $n-Y(i)< dM_i$, so
$m_j,n_j\le 2dM_i$ in $c_n$ for all $i\le j\le d$.
\end{lemma}
\begin{proof}
We have that $Y(i)-X(i)\le dM_i$ implies that $val(\pi[1,Y(i)],x_i)\le dM_i+1$, and since 
$val(\pi[1,Y(i)+k],y_i)\ge 2^k$ we get that:
\begin{itemize}
\item $val(\pi,y_i)\ge 2^{n-Y(i)};$ and
\item $val(\pi,x_i)\le dM_i+n-Y(i)+1.$
\end{itemize}
Assume that $n-Y(i)\ge dM_i$. Then, $2^{n-Y(i)} - (dM_i+n-Y(i)+1)> 2^{n-Y(i)} - (2n-2Y(i)+1) >0$, if $n-Y(i)\ge 3$. However, $\pi$ is an accepting path, so $val(\pi,x_i)=val(\pi,y_i)$, and we get a contradiction. Thus, we must have that $n-Y(i)< dM_i$ which implies that $val(\pi,x_i)\le 2dM_i$, so $m_i=n_i\le 2dM_i$ and for any 
$j$, $i< j\le d$, $m_j\le m_i$ and $n_j\le n_i$, so $m_j,n_j< 2M_i$ 
in $c_n$, for all $i\le j\le d$ since $\pi$ is an accepting path.
\end{proof}

Let $R_0 \defeq f_V(\pi,1,1,\ldots,1)$. Note that $R_0$ together with
$X$ and $Y$ as defined above adheres to Conditions~(a)--(e) of
reachability certificates.
%% It remains to either identify that
%% $\ell_1=\top$ or to identify a suitable loop in $R_0$. We distinguish
%% two cases:
%% %
%% \begin{itemize}
%% \item $Y(1)-X(1) \le d\cdot M_1$: we choose $\ell_1 \defeq \top$ and
%%   $L(1)\defeq X(1)$.
%%   %% \Cref{lm:smallpath} guarantees that
%%   %% $(f_V(\pi, \ell_1,\ldots,\ell_d),X,Y,L)$ is a witnessing
%%   %% certificate.
%% \item $Y(1)-X(1) > d\cdot M_1$: then
%%   there is a segment in $R_0[X(1),Y(1)]$ of length greater than $M_1$
%%   on which no $x$-counter has its first $++$ transition. Hence, there is a
%%   smallest $k$, $X(1)\le k<Y(1)$ and a $\alpha_k \xrightarrow{t_{k}}
%%   \cdots \xrightarrow{t_{k+\ell}} \alpha_{k+\ell+1}=\alpha_k$ in $R_0$
%%   since $M_1>\abs Q$. In particular $k+\ell+1\le Y(1)$. We choose
%%   $L(1)\defeq k$ and $\ell_1 \defeq \ell$.
%% \end{itemize}
%% %
%% We define $R_1 \defeq f_V(\pi,\ell_1,1,\ldots,1)$.

Suppose $R_{i-1}$ and $\ell_1,\ldots,\ell_{i-1}$ have been
constructed. If $i>1$, and $L(i-1)\ge X(i)$ or $\ell_{i-1} = \top$
then we choose $\ell_{i} \defeq \ell_{i-1}$, $L(i)=L(i-1)$ and
$R_i \defeq f_V(\pi,\ell_1,\ldots,\ell_{i},1,\ldots,1)$. Otherwise, we
distinguish two cases.
\begin{itemize}
\item $Y(i)-X(i) < dM_i$: we choose $\ell_{i} \defeq \top$ and $L(i)\defeq X(i)$.
\item $Y(i)-X(i) \ge dM_i$: then there is a segment in $R_{i-1}[X(i),Y(i)]$ of length greater than $M_i$ on which no $x$-counter has its first $\wPP$ transition. Let $N_i$ be the maximum number of different abstract configurations on this segment. Since $\ell_{i-1}\neq\top$ we know that $m_j,n_j$ can take at most $M_j$ different values for all $1\le j<i$, as they can either be $\bot$ or a residue class modulo $M_j$. Also, for all $i\le j\le d$ the values of $m_j,n_j$ have a constant value, either 0 or $\bot$,
on this segment, and $u_i=\cdots=u_d=0$ in all abstract configurations
of this segment. So
\begin{align*}
  N_i \le{} & \abs Q \prod_{1\le j<i} M_j^2 \cdot U_j\\
  \le{} & \abs Q \prod_{1\le j<i} \abs Q^{(1/4)\cdot 32^{j-1} + 2 + 32^{j-1} + 4 }\\
  \le{} & \abs Q^{(1/3968)(5\cdot 32^i - 23968)+6i+1}\\
  <{} & \abs Q^{(1/8)\cdot 32^{i-1} + 1}\\
  ={} & M_i 
\end{align*}
By the pigeonhole principle, there is a smallest $k$, $X(i)\le k<
Y(i)$, $\ell < M_i $, and a simple loop
$\alpha_k \xrightarrow{t_k}\cdots \xrightarrow{t_{k+\ell}}\alpha_{k+\ell+1}=\alpha_k$
in $R_{i-1}$. We choose $L(i)\defeq k$, $\ell_i\defeq \ell$ and let
$R_i\defeq f_V(\pi,\ell_1,\ldots,\ell_{i},1,\ldots,1)$.
\end{itemize}
By construction, $(R_d,X,Y,L)$ is a reachability certificate. It
remains to turn it into a witnessing certificate. In particular, this
requires to removes from $R_d$, to ensure that the final segment of
$R_d$ is short, and to establish that $R_d$ is consistent with the
implied $y$-constraints.

Let $R \defeq R_d = f_V(\pi,\ell_1,\ldots,\ell_{d})$ and $a$ be the
largest index such that $\ell_a\neq \top$. In order to make $R$
loop-free, we iterate the following process:
\begin{itemize}
\item identify the first simple loop $\alpha_k \xrightarrow{t_k} \cdots \xrightarrow{t_{k+\ell}}
  \alpha_{k+\ell+1}$ in $R[1,Y(a)]$ and replace it by $\alpha_k$;
  observe that for $I \defeq \{k+1,\ldots,k+\ell\}$, we have $I \cap
  \{ X(i), Y(i), L(i) : 1\le i \le d \} = \emptyset$ since
  $\alpha_{X(i)-1} \xrightarrow{t_{X(i)-1}} \alpha_{X(i)}$ occurring in $R_d$ means
  that $x_i$ has value $\bot$ in $\alpha_{X(i)-1}$ and a value different from $\bot$ in
  $\alpha_{X(i)}$, and thus $\alpha_{X(i)}$ cannot be part of a
  loop; the same argument applies to any $Y(i)$. Finally, since $L(i)$
  was chosen as the index of the first configuration of the first
  cycle appearing after $X(i)$, we have $L(i)\not\in I$ for all $1\le
  i\le d$ as well.
\item update $X,Y,L$ such that for all $i$ such that $X(i)>k$,
  $X(i)\defeq X(i)-\ell$, and analogously $Y(i) \defeq Y(i) - \ell$
  and $L(i) \defeq L(i) - \ell$ for the respective $i$.
\end{itemize}
This process guarantees that $R[1,Y(a)]$ is loop-free. It is easy
to verify that $(R,X,Y,L)$ obtained in this way is a reachability
certificate and that the last abstract configuration of $R$ is
accepting.

We now show that the $y$-constraints induced by $R$ are valid in the
final configuration of $R$. To this end, we first show that for all
$1\le i\le d$ such that $\ell_i\neq\top$, $X(i+1) - X(i) <
U_i$. Consider the simple path $\alpha_{X(i)}
\xrightarrow{t_{X(i)}} \alpha_{X(i)+1} \xrightarrow{t_{X(i)+1}} \cdots \xrightarrow{t_{X(i+1)-1}}
\alpha_{X(i+1)}$. If $Y(i) \ge X(i+1)$ then clearly $X(i+1)-X(i) \le
N_i<U_i$, where $N_i$ is defined as above. Otherwise, there is a
$k\in\N$ such that the path decomposes as
\[
\alpha_{X(i)} \xrightarrow{t_{X(i)}}\cdots \alpha_{Y(i)}\xrightarrow{t_{Y(i)}}\cdots 
\alpha_{Y(i)+k}\xrightarrow{t_{Y(i)+k}}\cdots \xrightarrow{t_{X(i+1)-1}}\alpha_{X(i+1)}
\]
and
\begin{itemize}
\item $u_i=0$ in all abstract states $\alpha_j$ with $X(i)\le j\le
  Y(i)$;
\item $u_i=U_i$ in all abstract states $\alpha_j$ with $Y(i)+k\le j\le
  X(i+1)$; and
\item $k\le \log{U_i}$.
\end{itemize}

Thus, the maximum length of $R[X(i),X(i+1)]$ is bounded by:
\begin{align*}
  & N_i \cdot M_i + \log U_i + N_i \cdot 2M_i\\
  \le {} & 2 \cdot M_i^3 + \log U_i\\
%  \le {} & \abs Q^{(6/7)\cdot 32^{i-1}+4} + \abs Q^{5(i-1)+1}+4\abs Q\\
  \le {} & \abs Q^{(3/8)\cdot 32^{i-1}+4} + \abs Q^{5(i-1)+1}+4\abs Q\\
  < {} & \abs Q^{32^{i-1}+4}\\
  = {} & U_i
\end{align*}

We can now show that $R$ respects the induced $y$-constraints. Fix
some $1\le i\le d$ such that $\ell_i\neq \top$. We distinguish two
cases:
\begin{itemize}
\item  There is no $1\le j\le a$ such that $X(i)\le L(j)\le X(i+1)$. Thus, we know that $y_i-y_{i+1}=\delta_i$ is in the set of induced $y$-constraints. Also, $val(\pi,y_i)-val(\pi,y_{i+1}) = val(\pi,x_i)-val(\pi,x_{i+1}) = X(i+1)-X(i) = \delta_i$ since we did not remove any abstract loops on the segment of $R_d$ between the first $\wPP$ update for $x_i$ and the first $\wPP$ update for $x_{i+1}$. Finally, since $\delta_i<U_i$ by the above argument, we conclude that $u_i=\delta_i$ in the last abstract
configuration $R[n]$ of $R$. 
\item Otherwise, $X(i)\le L(i)\le Y(i)$, so $y_i-y_{i+1}\ge\delta_i$ 
is in the set of induced $y$-constraints. However, $val(\pi,y_i)-val(\pi,y_{i+1}) = val(\pi,x_i)-val(\pi,x_{i+1}) \ge X(i+1)-X(i) = \delta_i$ and again because $\delta_i<U_i$ we can conclude that $u_i\ge \delta_i$ in $R[n]$.
\end{itemize}
This establishes that the $y$-constraints are satisfied.  Let $n$ be
the index of the last abstract configuration of $R$. For the final
step, we now argue that $val(\pi(R),x_a)\le val(\pi(R),y_a)$ and $n-Y(a)\le
2dM_{d+1}$. We make a case distinction:
\begin{itemize}
\item $a=d$: Note that $val(\pi,x_d)=val(\pi,y_d)$.  Since we only
  remove loops from $R_d[1,Y(d)]$, we have that $val(\pi(R),x_d)\le
  val(\pi(R),y_d)$. If $n-Y(d)\le 2dM_{d+1}$ we are done with $(R,X,Y,L)$ as
  a witnessing certificate.  Otherwise, assume $n-Y(d)>
  2dM_{d+1}$. This implies that the path $R[Y(d),n]$ must contain at
  least one simple loop. Consider iterating the following process:
\begin{itemize}
\item remove the first simple loop from $R[Y(d),n]$ and update
  $n\defeq n-\ell$, where $\ell$ is the length of the simple loop that
  was removed; and                 
\item stop if $n-Y(d)\le 2dM_{d+1}$.
\end{itemize}
We argue that, $n-Y(d)\ge M_d^2$. Let $R'$ and $n'$ 
be the previous values of $R,n$ before the last iteration. It must be that $n'
>2dM_{d+1}$ and since the length of 
any simple loop of $R'[Y(d)+1,n']$ is bounded by 
$M_{d+1}$, we get that $n-Y(a)\ge M_{d+1}\ge M_d^2$. 
Note that $Y(d)-X(d)\le M_d\cdot \abs{Q}\cdot\prod_{1\le j< d}M_j^2\cdot U_j\le M_d^2$, 
so $val(\pi(R[1,Y(d)]),x_d)\le M_d^2$. It must be then 
the case that $val(\pi(R),x_d)\le val(\pi(R),y_d)$.

\item $a<d$: we know $n-Y(a)\le 2dM_{d+1}$ by \Cref{lm:smallpath}. Moreover, we must have that $val(\pi(R),x_a)\le val(\pi(R),y_a)$ since $val(\pi,x_a)=val(\pi,y_a)$ and we do not remove loops after the counter $y_a$ is incremented. 

\end{itemize}

\section{A decidable fragment of string constraints}

In this section, we show that a certain fragment of string constraints
whose decidability status has been left open in the literature can be
reduced in logarithmic space to generalised Sem{\"e}nov arithmetic,
and is hence decidable in EXPSPACE. This demonstrates an important
application of our results on generalised Sem{\"e}nov arithmetic, with
deep connections to solving string constraints in practice, which has
been one of the motivations for our work.

Let $\Sigma = \{ 0,1\}$. The \emph{theory of enriched string
  constraints} $\TRElnc$ is the first-order theory of the two-sorted
structure \[\langle \Sigma^*,\N; \{w\}_{w \in \Sigma^*}, \cdot, \slen,
\strnum, \{R_i \}_{i\in \N}, 0,1,+\rangle,\] where
\begin{itemize}
\item the binary function $\cdot$ over $\Sigma^*$ is the string
  concatenation operator,
\item the unary function $\slen\colon \Sigma^* \to \N$ returns on
  input $w$ the length $\abs{w}$ of $w$,
\item the unary function $\strnum\colon \Sigma^* \to \N$ on input $u$
  returns $\eval{u}$, and
\item $R_0,R_1,\ldots \subseteq \Sigma^*$ is an enumeration of all
  regular languages.
\end{itemize}
The remaining predicates, constant and function symbols are defined in
their standard semantics.

The above theory was introduced in~\cite{BKMMDNG20}, where an SMT solver addressing some fragments of this theory was defined, implemented, and compared to other state of the art solvers which can handle such string constraints. Extending \cite{BKMMDNG20}, \cite{BDGKMMN23} presents in more details the motivation behind considering this theory and its fragments. More precisely, the authors of \cite{BDGKMMN23} analysed an extensive collection of standard real-world benchmarks of string constraints and extracted the functions and predicates occurring in them. The works \cite{BKMMDNG20,BDGKMMN23} focused on benchmarks that do not contain word equations, and the result of the aforementioned benchmark-analysis produced exactly the four functions and predicates mentioned above: $\slen$, $\strnum$, regular language membership, and concatenation of strings. 

Complementing the practical results of \cite{BKMMDNG20}, \cite{BDGKMMN23} showed a series of theoretical results regarding fragments of $\TRElnc$. In particular, the existential theory of $\TRElnc$ is shown to be undecidable. Moreover, \cite{BDGKMMN23} leaves as an open problem the question whether the existential theories of $\TREln$ and $\TREnc$, which drop the concatenation operator and length function, respectively, are decidable. From these two, the existential theories of $\TREln$ seems particularly interesting, as all instances from the benchmarks considered in the analysis \cite{BDGKMMN23} can be easily translated into a formula from this particular fragment of $\TREln$. Indeed, by the results reported in Table 1.b from \cite{BDGKMMN23},  %(which we have also double checked), 
no instance contains both concatenation of strings and the $\strnum$ function; moreover, the concatenation of strings, which appears only in formulas involving regular membership predicates and, in some cases, length function, can be easily removed in all cases by a folklore technique called automata splitting (see, e.g., \cite{AACHRRS15}). Therefore, showing that the existential fragment of $\TREln$ is decidable essentially shows that one can decide all the instances from the standard benchmarks analysed in \cite{BDGKMMN23}.

In this paper, we solve this open problem. By a reduction to generalised Sem{\"e}nov arithmetic, we can settle the decidability status of $\TREln$:
\begin{theorem}
  The existential fragment of $\TREln$ is decidable in EXPSPACE.
\end{theorem}
%
%% We equivalently interpret both the unary functions $\slen$
%% and $\strnum$ and the binary function $+$ as binary and ternary
%% relations over $\Sigma^*\times\N$ and $\N^3$, respectively.

Again, we treat $\TREln$ as a relational structure. Without loss of
generality, we may assume that atomic formulas of $\TREln$ are one of
the following:
\begin{itemize}
\item $R(s)$ for some string variable $s$ and a regular language $R$;
\item $s=t$ for some string variables $s$ and $t$;
\item $\slen(s,x)$ or $\strnum(s,x)$ for some string variable $s$ and
  integer variable $x$; or
\item $\vec a \cdot \vec x \ge b$ for a vector of integer variables
  $\vec x$.
\end{itemize}
The size of a formula of $\TREln$ is defined in the standard way as
the number of symbols required to write it down, assuming binary
encoding of numbers, and where the size of some $R$ is the size of the
smallest DFA accepting $R$. Furthermore, in a quantifier-free formula
$\varphi$ of $\TREln$, we may without loss of generality assume that
all atomic formulas occur positive, except for atomic formulas $s=t$.

We now describe the reduction to existential Sem{\"e}nov
arithmetic. The idea underlying our proof is that we map a string $s$
to the number $\eval{1s}$. Note that we cannot directly treat strings
in $\Sigma^*$ as natural numbers due to the possibility of leading
zeros. This encoding enables us to treat strings as numbers and to
implement the functions $\strnum$ and $\slen$ in generalised
Sem{\"e}nov arithmetic. Given a quantifier-free formula $\varphi$ of
$\TREln$, we define by structural induction on $\varphi$ a function
$\sigma$ that maps $\varphi$ to an equi-satisfiable formula of
generalised Sem{\"e}nov arithmetic:
\begin{itemize}
\item Case $\varphi \equiv R(s)$: $\sigma(\varphi)\defeq (0^*1 R)(s)$;
\item Case $\varphi \equiv s=t$ or $\varphi \equiv \neg(s=t)$:
  $\sigma(\varphi) \defeq s=t$ or $\sigma(\varphi) \defeq \neg s=t$,
  respectively;
\item Case $\varphi \equiv \strnum(s,x)$: $\sigma(\varphi)\defeq
  \exists y.\, 2^y\leq s\land s< 2^{y+1} \land x = s-2^y$;
\item Case $\varphi \equiv \slen(s,x)$: $\sigma(\varphi)\defeq 2^{x}\le s\land s < 2^{x+1}$;
\item Case $\varphi \equiv \vec a\cdot \vec x \ge b$: $\sigma(\varphi) \defeq \vec a\cdot \vec x \ge b$; and
\item Case $\varphi \equiv \varphi_1 \sim \varphi_2$, ${}\sim{} \in \{
  \wedge, \vee \}$: $\sigma(\varphi)\defeq \sigma(\varphi_1) \sim
  \sigma(\varphi_2)$.
\end{itemize}

\begin{lemma}
  Let $\varphi$ be a quantifier-free formula of $\TREln$ and $S$ be
  the set of string variables occurring in $S$. Then $\varphi$ is
  satisfiable if and only if $\sigma(\varphi) \wedge \bigwedge_{s\in
    S} s>0$ is satisfiable.
\end{lemma}
\begin{proof}
  Observe that the variables occurring in $\varphi$ are the same
  variables as those occurring in $\sigma(\varphi)$. Let $S$ be the
  set of string variables in $\varphi$ and $X$ be the set of
  integer-valued variables in $\varphi$. Given an assignment $\mathcal
  I_S \colon S \to \{0,1\}^*$, we define $\tilde{\mathcal{I}}_S \defeq
  S \to \N$ such that $\tilde{\mathcal I}_S(s) \defeq \eval{1 \mathcal
    I_S(s)}$. Subsequently, denote by $\mathcal I_X\colon X \to \N$ an
  assignment to the integer-valued variables. We show by structural
  induction on $\varphi$ that $(\mathcal I_S,\mathcal I_x) \models
  \varphi$ if and only if $(\tilde{\mathcal I}_S, \mathcal I_X)
  \models \sigma(\varphi) \wedge \bigwedge_{s\in S} s>0$:
  \begin{itemize}
  \item Case $\varphi \equiv R(s)$: Let $\mathcal I_S(s)=b_{n-1}\cdots
    b_0$, we have $\mathcal I_S(s)\in R$ if and only if $2^n +
    \sum_{i=0}^{n-1} 2^ib_i \in \eval{0^*1R}$, noting that $2^n +
    \sum_{i=0}^{n-1}2^i b_i=\eval{1b_{n-1}\cdots b_0}=\tilde{\mathcal
      I}_S(s)$.
  \item Case $\varphi \equiv \strnum(s,x)$: Let $\mathcal
    I_S(s)=b_{n-1}\cdots b_0$ and $\mathcal I_X(x)=m$. We have that
    $m=\sum_{i=0}^{n-1}2^ib_i$ if and only if $m=\tilde{\mathcal
      I}_S(s)-2^n$ if and only if $(\tilde{\mathcal I}_S,\mathcal
    I_X)\models \sigma(\varphi) \wedge \bigwedge_{s\in S} s>0$.
  \item Case $\varphi\equiv \slen(s,x)$: Let $\mathcal
    I_S(s)=b_{n-1}\cdots b_0$ and $\mathcal I_X(x)=m$. We have that
    $m=n$ if and only if $2^m\le \eval{1b_{n-1}\cdots b_0} <2^{m+1}$
    if and only if $(\tilde{\mathcal I}_S,\mathcal I_X)\models
    \sigma(\varphi) \wedge \bigwedge_{s\in S} s>0$.
  \end{itemize}
  The remaining cases follow obviously.
\end{proof}

\section{Conclusion}

The main result of this article has been to show that the existential
theory of generalised Sem{\"e}nov arithmetic is decidable in
EXPSPACE. As an application of this result, we showed that a highly
relevant class of string constraints with length constraints is also
decidable in EXPSPACE; the decidability of this class was the main
problem left open in~\cite{BDGKMMN23}. On a technical level, those
results were obtained by showing that a restricted class of labelled
affine VASS has an EXPSPACE-decidable language emptiness problem. The
structural restrictions imposed on those restricted \laVASS are rather
strong, though necessary to obtain a decidable class of \laVASS.

An interesting aspect of our approach is that it establishes
automaticity of the existential fragment of a logical theory that is
different from traditional notions of automaticity, which are based on
finite-state automata or tree automata over finite or infinite words
and trees~\cite{KN95,BG00}, respectively. It would be interesting to
better understand whether there are natural logical theories whose
(existential) fragments are, say, Petri-net or visibly-pushdown
automatic.

We have ignored algorithmic lower bounds throughout this article, but
it would, of course, be interesting to see whether the upper bounds of
the decision problems we considered in this article are tight. It is
clear that generalised Sem{\"e}nov arithmetic is PSPACE-hard since it
can readily express the DFA intersection non-emptiness problem, but
this still leaves a considerable gap with respect to the EXPSPACE
upper bound we established. In particular, the recent results
of~\cite{BCM23} showing an NEXP upper bound for the existential
fragment of Sem{\"e}nov arithmetic suggest that, if an EXPSPACE lower
bound for existential generalised Sem{\"e}nov arithmetic is possible,
it will require the use of regular predicates.

\bibliography{bibliography}

\newpage

\onecolumn

\appendix
\section{Closure properties of \laVASS languages}

%\subsection{Closure properties of \laVASS languages}\label{app:intersection}

Let $V_i=\langle Q_i,d_i, \Sigma,\Delta_i,\lambda_i,q_0^{(i)},F_i,\phi_i \rangle$,
$i\in
\{1,2\}$, be two \laVASS.

\begin{proposition}
  The languages of \laVASS are closed under union and intersection.
\end{proposition}

Closure under union is trivial since we allow for non-determinism. To
show closure under intersection, we define the \laVASS $V \defeq
(Q',d_1+d_2,\Sigma,\Delta',\lambda',q_0',F',\phi')$ such that
\begin{itemize}
\item $Q' \defeq Q_1\times Q_2$,
\item $((q_1,q_2),a,(r_1,r_2)) \in \Delta'$ if and only if $(q_1,a,r_1)\in \Delta_1$ and $(q_2,a,r_2)\in \Delta_2$,
\item $\lambda'((q_1,q_2),a,(r_1,r_2)) \defeq (\lambda_1(q_1,a,r_1), \lambda_2(q_2,a,r_2))$,
\item $q_0'' \defeq (q_0^{(1)},q_0^{(2)})$,
\item $F' \defeq F_1\times F_2$, and
\item $\phi$ is the conjunction of $\phi_1$ and $\phi_2$, with counters
renamed accordingly.
\end{itemize}

\begin{lemma}\label{lem:intersect-reach}
  For any $w\in \Sigma^*$, $q^{1},r^{(1)}\in Q_1$ and $q^{2},r^{2}\in
  Q_2$, the following are equivalent:
  \begin{enumerate}[(i)]
  \item
    $((q^{(1)},q^{(2)}),m_1,\dots,m_{d_1+d_2})\xrightarrow{w}_{V}
    ((r^{(1)}, r^{(2)}),m_1',\dots,m_{d_1+d_2}')$
  \item
  $(q^{(1)},m_1,\dots,m_{d_1})\xrightarrow{w}_{V_1}
    (r^{(1)},m_1',\dots,m_{d_1}')$ and
    $(q^{(2)},m_{d_1+1},m_{d_1+1}\dots,m_{d_1+d_2})\xrightarrow{w}_{V_2}
    (r^{(1)},m_{d_1+1}',\dots,m_{d_1+d_2}')$.
  \end{enumerate}
\end{lemma}
\begin{proof}
  Let $w\in \Sigma^*$, we prove the statement by induction on $\abs
  w$. The base case $w=\epsilon$ is immediate by the definition of
  $V$.

  For the induction step, let $w=u\cdot a$ for some
  $a\in\Sigma$. The induction hypothesis yields
  \begin{align*}
    & ((q^{(1)},q^{(2)}),m_1,\dots,m_{d_1+d+2})\xrightarrow{u}_V
    ((s^{(1)}, s^{2}),m_1'',\dots,m_{d_1+d_2}'')\\
    \iff & (q^{(1)},m_1,\dots,m_{d_1})\xrightarrow{u}_{V_1}
    (s^{(1)},m_1'',\dots,m_{d_1}'') \text{ and }\\
    & (q^{(2)},m_{d_1+1},\dots,m_{d_1+ d_2})\xrightarrow{u}_{V_2}
    (s^{(2)},m_{d_1+1}'',\dots,m_{d_1+d_2}'')\,.
  \end{align*}
  Again, by definition of $V$ we furthermore have
  \begin{align*}
    & ((s^{(1)},s^{(2)}),m_1'',\dots,m_{d_1+d+2}'')\xrightarrow{a}_V
    ((r^{(1)}, r^{2}),m_1',\dots,m_{d_1+d_2}'')\\
    \iff & (s^{(1)},m_1,\dots,m_{d_1})\xrightarrow{a}_{V_1}
    (r^{(1)},m_1'',\dots,m_{d_1}'') \text{ and }\\
    & (s^{(2)},m_{d_1+1},\dots,m_{d_1+ d_2})\xrightarrow{a}_{V_2}
    (r^{(2)},m_{d_1+1}'',\dots,m_{d_1+d_2}'')\,.
  \end{align*}
  This concludes the proof of the statement.
\end{proof}

\begin{corollary}
  Let $V_1,V_2$ be \laVASS. Then $L(V_1)\cap L(V_2))=L(V_1\cap V_2)$.
\end{corollary}
\begin{proof}
  For any $w\in\Sigma^*$, by \Cref{lem:intersect-reach}, we have:
  \begin{align*}
    & w\in L(V_1\cap V_2)\\
    \iff & (q_0',0\dots,0)\xrightarrow{w}_V((r_1,r_2),m_1,\dots,m_{d_1+d_2}) \text{
      for some } (r_1,r_2)\in F', (m_1,\dots,m_{d_1+d_2})\in S_f'\\
    \iff & (q_0^{(1)},0,\dots,0)\xrightarrow{w}_{V_1}
    (r_1,m_1,\dots,m_{d_1}) \text{ for some } r_1\in F_1, (m_1,\dots,m_{d_1})\in
    S_f^{(1)}, \text{ and }\\
    & (q_0^{(2)},0,\dots,0)\xrightarrow{w}_{V_2} (r_2, m_{d_1+1},\dots,m_{d_1+d_2}) \text{
      for some } r_2\in F_2\, (m_{d_1+1},\dots, m_{d_1+d_2})\in
    S_f^{(2)}\\
    \iff & w \in L(V_1) \text{ and } w\in L(V_2)\,.\qedhere
  \end{align*}
\end{proof}

From the construction, it is clear that for the size of the \laVASS
for $L(V_1)\cap L(V_2)$, we have:
\begin{itemize}
\item $\abs Q = \abs{Q_1} \cdot \abs{Q_2}$,
\item $\abs \Delta \le \abs{\Delta_1}\cdot \abs{\Delta_2}$, 
\item $\abs \lambda = max(\abs{\lambda_1}, \abs{\lambda_2})$, and
\item $d=d_1+d_2$.
\end{itemize}

\end{document}